\documentclass[12pt]{amsart}
\usepackage{tikz}
\usepackage{amsmath}
\usepackage{graphicx}
\usetikzlibrary{arrows}
\usepackage[margin=1in, footskip=0.3in]{geometry}

\usepackage{algorithm}
\usepackage{algpseudocode}
\usepackage{subfigure}
\usepackage{hyperref}

\newtheorem{theorem}{Theorem}[section]
\newtheorem{lemma}[theorem]{Lemma}

\newtheorem{definition}[theorem]{Definition}
\newtheorem{remark}[theorem]{Remark}

\def\R{\mathbb{R}}
\def\x{x}
\def\zero{\vec{0}}
\def\hatd{\hat{D}}
\def\Di{D^i}
\def\Distar{D^i_*}
\def \S {S}
\def \s {s}
\def \h {s}
\newcommand{\fullversion}[1]{}

\begin{document}

\title{Algorithms for Colourful Simplicial Depth and Medians in the Plane}
\author{Olga Zasenko}
\email{ozasenko@sfu.ca}
\author{Tamon Stephen}
\email{tamon@sfu.ca}
\address{Simon Fraser University, British Columbia, Canada}

\maketitle
\pagenumbering{arabic}

\begin{abstract}
The {\it colourful simplicial depth} (CSD) of a point $\x \in \R^2$ relative
to a configuration $P=(P^1, P^2, \ldots, P^k)$ of $n$ points in $k$ colour 
classes is exactly the number of closed simplices (triangles) with vertices from 
3 different colour classes that contain $\x$ in their convex hull. 
We consider the problems of efficiently computing the colourful simplicial
depth of a point $x$, and of finding a point in $\R^2$, called a {\it median},
that maximizes colourful simplicial depth.

For computing the colourful simplicial depth of $\x$, our algorithm runs in
time $O\left(n \log{n} + k n \right)$ in general, and $O(kn)$ if the points 
are sorted around $\x$.  For finding the colourful median, we get a time
of $O(n^4)$.  For comparison, the running times of the best known algorithm
for the monochrome version of these problems are $O\left(n \log{n} \right)$ 
in general, improving to $O(n)$ if the points are sorted around $\x$
for monochrome depth, and $O(n^4)$ for finding a monochrome median. 
\end{abstract}

\section{Introduction}
\label{S:1}
The \textit{simplicial depth} of a point $\x \in \R^2$ relative to a set $P$ of
$n$ data points is exactly the number of simplices (triangles) formed with the
points from $P$ that contain $\x$ in their convex hull. A
\textit{simplicial median} of the set $P$ is any point in $\R^2$ which is
contained in the most triangles formed by elements of $P$,
i.e.~has maximum simplicial depth with respect to $P$.
Here we consider a set $P$ that consists of $k$ colour classes
$P^1, \ldots, P^k$. The {\it colourful simplicial depth} 
of $\x$ with respect to configuration $P$
is the number of triangles with vertices from 3 different colour classes
that contain $\x$. A {\it colourful simplicial median} of a configuration
$P=(P^1, P^2, \ldots, P^k)$ is any point in the convex hull of $P$
with maximum colourful simplicial depth. 

The monochrome simplicial depth was introduced by Liu \cite{Liu90}.
Up to a constant,
it can be interpreted as the probability that $x$ is in the convex hull
of a random simplex generated by $P$.  The colourful version, 
see \cite{MR2225675}, 
generalizes this to selecting points from $k$ distributions.  
Then medians are central points which are in some sense most 
representative of the distribution(s).  
Our objective is to find efficient algorithms for finding both the
colourful simplicial depth of a given point $\x$ with respect to a
configuration, and a colourful simplicial median of a configuration.

\subsection{Background}\label{S:background}
Both monochrome and colourful simplicial depth extend to $\R^d$ and 
are natural objects of study in discrete geometry.  For more background
on simplicial depth and competing measures of data depth, see
\cite{Alo06} and \cite{FR05}.  Monochrome depth has seen a flurry of
activity in the past few years, most notably relating to the {\it First
Selection Lemma}, which is a lower bound for the depth of the median,
see e.g.~\cite{MW14}.  

The colourful setting for simplicial depth is suggested by
B{\'a}r{\'a}ny's approach \cite{MR676720} to proving a colourful
version of Carath{\'e}odory's theorem. Deza et al.~\cite{MR2225675}
formalized the notion and considered bounds for the colourful depth of
points in the intersection of the convex hulls of the colours.
Among the recent work on colourful depth are proofs of the
lower~\cite{Sar15} and upper~\cite{ABP+16} bounds conjectured by
Deza et al., with the latter result showing beautiful connections
to Minkowski sums of polytopes.

The monochrome simplicial depth can be computed by enumerating 
simplices, but in general dimension, it is quite challenging to
compute it more efficiently~\cite{Alo06}, \cite{CO01}, \cite{FR05}.
Several authors have considered the two-dimensional version of the
problem, including Khuller and Mitchell \cite{MR1045522}, Gil, Steiger 
and Wigderson \cite{MR1189827} and Rousseeuw and Ruts
\cite{rousseeuw1996algorithm}. Each of these groups produced an
algorithm that computes the monochrome depth 
in $O(n\log{n})$ time, with sorting the input as the bottleneck. 
If the input points are sorted, these algorithms take linear time.

For simplicial medians, Khuller and Mitchell~\cite{MR1045522}, and 
Gil, Steiger, and Wigderson~\cite{MR1189827} 
considered an \textit{in-sample} version of a simplicial 
median, that, they looked for a point \emph{from $P$} with maximum simplicial 
depth.  However, we consider a \textit{simplicial median} to be {\it any} 
point $\x \in \R^2$ maximizing the simplicial depth.  
Rousseeuw and Ruts~\cite{rousseeuw1996algorithm} found an algorithm to
compute the simplicial median in $O(n^5 \log{n})$ time,
Aloupis et al.~\cite{MR1989273} improved this to $O(n^4)$.
This is arguably as good as should be expected, following
the observation of Lemma~\ref{L:intersection_points} in Section~\ref{SS:ccsm:1}
that shows that there are in some sense $\Theta(n^4)$ candidate points
for the location of the colourful median.

\begin{remark}
The two groups \cite{MR1045522,MR1189827} who studied the in-sample median 
computed the simplicial depth of each point in the data set $P$ in total
$O(n^2)$ time. 
\end{remark}

\subsection{Organization and Main Results}
In Section~\ref{S:ccsd}, we develop an algorithm for computing
colourful simplicial depth that runs in $O(n \log{n} + kn)$ time. This
retains the $O(n\log{n})$ asymptotics of the monochrome algorithms when 
$k$ is fixed. As in the monochrome case, sorting the initial input is a 
bottleneck, and the time drops to $O(kn)$ if the input is sorted around $\x$.  
In this case, for fixed $k$, it is a linear time algorithm.

In Section~\ref{S:ccsm}, we turn our attention to computing a
colourful simplicial median.  We develop an algorithm that does
this in $O(n^4)$ time using a topological sweep.  This is 
independent of $k$ and matches the running time from the monotone case.
Section~\ref{S:Concl} contains
conclusions and discussion about future directions.

\section{Computing Colourful Simplicial Depth}
\label{S:ccsd}
\subsection{Preliminaries}
\label{S:ccsdp}

We consider a family of sets $P^1$, $P^2$, $\ldots$ , $P^k$
$\subseteq \R^2$, $k \geq 3$, where each $P^i$ consists of the
points of some particular colour $i$. Refer to the $j^{th}$ element of
$P^i$ as $P^i_j$. 
We generally use superscripts for colour classes, while 
subscripts indicate the position in the array.
We will sometime perform arithmetic operations on the subscripts,
in which case the indices are taken modulo the size of the array
i.e.~(mod $n_i$). 

We denote the union of all colour sets by $P$:
$P = \bigcup\limits_{i = 1}^k P^i$. The total number of points is
$n$, where $\vert P^i \vert = n_i$, $\sum\limits_{i = 1}^k n_i = n$. 
We assume that points of $P \bigcup \lbrace\x\rbrace$ are
in general position to avoid technicalities. Without loss of
generality, we can take $\x = {\zero}$, the zero vector.

\begin{definition}{} \label{def:col_tri}
A \emph{colourful triangle} is a triangle with one vertex of each
colour, i.e. it is a triangle whose vertices $v_1, v_2, v_3$ are
chosen from distinct sets $P^{i_1}$, $P^{i_2}$, $P^{i_3}$,
where $i_i \ne i_2, i_3; i_2 \ne i_3$.  
\end{definition}

\begin{definition}{} \label{def:csd}
The \emph{colourful simplicial depth} $\hatd(\x, P)$ of a point $\x$
relative to the set $P$ in $\R^2$ is the number of colourful
triangles containing $\x$ in their convex hull.  We reserve $D(\x,P)$ for
the (monochrome) simplicial depth, which counts all triangles
from $P$ regardless of the colours of their vertices.
\end{definition}

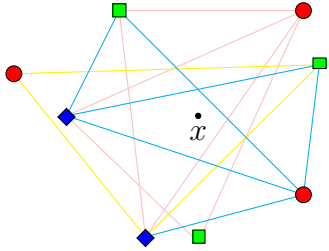
\begin{figure}[h]
\center
  \begin{tikzpicture}[scale=0.7]
    \filldraw (0,0) circle (0.05cm) node[below] {$x$};
    \coordinate(r1) at (2, 2);
    \coordinate(r3) at (2, -1.5);
    \coordinate(r2) at (-3.5, 0.8);
    \coordinate(g2) at (-1.5, 2);
    \coordinate(g1) at (2.3, 0.97);
    \coordinate(g3) at (0, -2.4);
    \coordinate(b2) at (-1, -2.3);
    \coordinate(b1) at (-2.5, 0);
    \draw[color=pink] (r1) -- (g2) -- (b2) -- cycle;
    \draw[color=pink] (r1) -- (g3) -- (b1) -- cycle;
    \draw[color=yellow] (r2) -- (g1) -- (b2) -- cycle;
    \draw[color=cyan] (r3) -- (g1) -- (b1) -- cycle;
    \draw[color=cyan] (r3) -- (g2) -- (b1);
    \draw[color=cyan] (r3) -- (b2);
    \filldraw[color=black, fill=red] (r1) circle (0.15cm);
    \filldraw[color=black, fill=red] (r2) circle (0.15cm);
    \filldraw[color=black, fill=red] (r3) circle (0.15cm) ;
    \filldraw[color=black, fill=green] (-1.37, 2.13) rectangle (-1.62, 1.88);
    \filldraw[color=black, fill=green] (2.43, 1.1) rectangle (2.18, 0.9);
    \filldraw[color=black, fill=green] (0.13, -2.17) rectangle (-0.1, -2.42);
    \filldraw[color=black, fill=blue, rotate around={45:(b2)}] (-1.13, -2.43) rectangle (-0.9, -2.2);
    \filldraw[color=black, fill=blue, rotate around={45:(b1)}] (-2.4, 0.1) rectangle (-2.63, -0.13);
\end{tikzpicture}
\caption[ ]{A configuration $P$ of 8 points in $\R^2$ surrounding a
  point $\x$ with $\hatd(\x,P)=6$.}
\label{fig:CSD_in_color}
\end{figure}

\begin{remark} We are checking containment in \emph{closed} triangles.
With our general position assumption, this will not affect the
value of $\hatd(\x,P)$.  It is more natural to consider closed triangles
than open triangles in defining colourful medians; 
the open triangles version of this question may also be interesting.
\end{remark}

Throughout the paper we work with polar angles $\theta^i_j$ formed by
the data points $P^i_j$ and a fixed ray from $\x$. 
We remark that simplicial containment does not change as points
are moved on rays from $\x$, see for example~\cite{MR1790005}.
Thus we can ignore the moduluses of the $P^i_j$, and work entirely
with the $\theta^i_j$, which lie on the unit circle $\mathcal{C}$ 
with $\x$ as its origin.  
We will at times abuse notation, and not distinguish between $P^i_j$ 
and $\theta^i_j$. 

Note that the ray taken to have angle 0 is arbitrary, and may be chosen 
based on an underlying coordinate system if available, or set to the
direction of the first data point $P_1$. 
We can sort the input by polar angle, in other words, we can order the 
points around $\x$. (Perhaps it is naturally presented this way.)
We reduce the $\theta^i_j$ to lie in the range $[0, 2\pi)$. 

The \textit{antipode} of some point $\alpha$ 
on the unit circle is $\bar{\alpha} =
(\alpha + \pi) \bmod{2\pi}$. A key fact  in computing CSD is that
a triangle $\bigtriangleup abc$ does 
\textit{not} contain $\x$ if and only if the corresponding polar
angles of points $a$, $b$ and $c$ lie on a circular arc of less than
$\pi$ radians. This is illustrated in Fig.~\ref{fig:lemma_antipodes}, and
is equivalent to the following lemma, stated by 
Gil, Steiger and Wigderson~\cite{MR1189827}:

\begin{lemma} \label{le:topological_Gil}
Given points $a$, $b$, $c$ on the unit circle $\mathcal{C}$ centred at
$\x$, let $\bar{a}$ be antipodal to $a$. Then $\bigtriangleup
abc$ contains $\x$ if and only if $\bar{a}$ is
contained in the minor arc (i.e.~of at most $\pi$ radians) with 
endpoints $b$ and $c$.
\end{lemma}
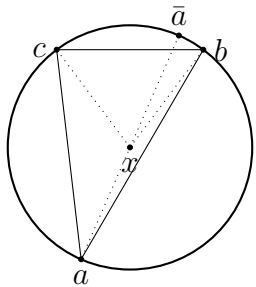
\begin{figure}[h]
\center
  \begin{tikzpicture}[scale=0.65]
    \draw[thick] (0,0) circle (2.5cm);
    \filldraw (0,0) circle (0.05cm) node[below] {$x$};
    \filldraw (1.5,2) circle (0.05cm) coordinate(b) node[right] {$b$};
    \filldraw (-1.5,2) circle (0.05cm) coordinate(c) node[left] {$c$};
    \filldraw (1,2.29) circle (0.05cm) node[above] {$\bar{a}$};
    \filldraw (-1,-2.29) circle (0.05cm) coordinate(a) node[below] {$a$};
    \draw[dotted] (1,2.29) -- (-1,-2.29);
    \draw[dotted] (0, 0) -- (1.5,2);
    \draw[dotted] (0, 0) -- (-1.5,2);
    \draw (a) -- (b) -- (c) -- cycle;
\end{tikzpicture}
\caption[ ]{Antipode $\bar{a}$ falls in the minor arc between $b$ and
$c$ and, therefore, the triangle $\bigtriangleup abc$ contains
$\x$.}
\label{fig:lemma_antipodes}
\end{figure}

\subsection{Outline of Strategy}
\label{SS:ccsds}

Recall that we denote the ordinary and colourful simplicial depth
by $D(\x,P)$ and $\hatd(\x,P)$ respectively. We can compute
$\hatd(\x,P)$ by first computing $D(x,P)$ and then removing all
triangles that contain less than three distinct colours. To this end, we
denote the number of triangles with at least two vertices of colour $i$ 
as $\Di(\x,P)$. When $\x$ and $P$ are clear from the context, we will
abbreviate these to $D$, $\hatd$ and $\Di$.

Since we can compute $D(\x,P)$ efficiently using the algorithms
mentioned in the introduction~\cite{MR1189827}, \cite{MR1045522},
\cite{rousseeuw1996algorithm}, the challenge is to compute 
$\Di(\x,P)$ for each $i=1,2,\ldots,k$. Then we conclude 
$\hatd(\x, P)= D(\x, P) - \sum\limits_{i=1}^k \Di(\x, P)$. 
To compute $\Di$ efficiently for each colour
$i$, we walk around the unit circle tracking the minor arcs
between pairs of points of colour $i$, and the number of antipodes
between them. We do this is in linear time in $n$ by
moving the front and back of the interval once around the circle, and
adjusting the number of relevant antipodes with each move. This
builds on the approach of Gil, Steiger and Wigderson~\cite{MR1189827}
for monochrome depth.

\begin{remark}
\label{R:distar} 
When computing $\Di$, we count antipodes of all $k$ colours;
the triangles with three vertices of colour $i$ will be counted three
times: $\bigtriangleup abc$, $\bigtriangleup bca$ and $\bigtriangleup cab$. 
Thus the quantity obtained by this count is in
fact $\Distar := \Di + 2 \sum\limits_{i = 1}^k D(\x, P^i)$. 
We separately compute $\sum\limits_{i = 1}^k D(\x, P^i)$, allowing us to
correct for the overcounting at the end.
\end{remark}

\subsection{Data Structures and Preprocessing}
\label{SS:ccsdm}
We begin with the arrays $\theta^i$ of polar angles, which we sort if
necessary. All elements in $\bigcup\limits_{i = 1}^k \theta^i$ are
distinct due to the general position requirement. 
By construction we have:
\begin{equation}
	0 \leq \theta^i_0 < \theta^i_1 < \ldots < \theta^i_{n_i - 1} < 2\pi, \quad\mbox{for all } 1 \leq i \leq k ~ .
\end{equation}
Let $\bar{\theta}^i$ be the array of antipodes of $\theta^i$, also
sorted in ascending order.  We generate $\bar{\theta}^i$ by 
finding the first $\theta^i_j \ge \pi$, moving the part of the array
that begins with that element to the front, and hence the front of
the original array to the back; $\pi$ is subtracted from the elements
moved to the front and added to those moved to the back.  This takes
linear time.

We merge all $\bar{\theta}^i$ into a common sorted array denoted by
$A$. Now we have all antipodes ordered as if we were
scanning them in counter-clockwise order around the circle
$\mathcal{C}$ with origin $\x$. Let us index the $n$ elements of $A$ 
starting from 0. Then, for each colour $i = 1, 
\ldots, k$, we merge $A$ and $\theta^i$ into a sorted array $A^i$.
Once again, this corresponds to a counter-clockwise ordering of 
data points 
around $\mathcal{C}$. 

While building $A^i$, we associate pointers from the elements of 
array $\theta^i$ to the corresponding position (index) in $A^i$. 
This is done by updating the pointers whenever a swap occurs during the
process of merging the arrays. Denote the index of some $\theta^i_j$
in $A^i$ by $p(\theta^i_j)$. Then the number of the antipodes that
fall in the minor arc between two consecutive points $\theta^i_j$
and $\theta^i_{j+1}$ on $\mathcal{C}$ is
$\left( p \left(\theta^i_{j+1} \right) - p\left(\theta^i_j \right) -
1 \right)$, if $p \left(\theta^i_j \right) < p \left(\theta^i_{j+1}
\right)$, or $\left(n + n_i - p \left(\theta^i_j \right) + p \left(
\theta^i_{j+1} \right) - 1\right)$, if $p\left( \theta^i_j \right) >
p \left( \theta^i_{j+1} \right)$. Note that $p\left( \theta^i_j \right)$
is never equal to $p\left( \theta^i_{j + 1} \right)$.

Now, for each point $\theta^i_j$, we find the index
$l(i, j)$ in the corresponding array $\theta^i$ such that
$\angle{\theta^i_j, x ,\theta^i_{l(i, j)}} < \pi$ and
$\angle{\theta^i_j, x, \theta^i_{l(i, j) + 1}} > \pi$
(Fig. \ref{fig:l(i,j) vs l(i,j)+1}). Thus the sequence of points
$\theta^i_{j}, \theta^i_{j+1}, \ldots, \theta^i_{l(i, j)}$ 
is maximal on an arc shorter than $\pi$. Viewing the minor arc between two
points as an interval, the intervals 
with left endpoint $\theta^i_j$ and right end point from this sequence
overlap and can be split into small disjoint intervals as follows:
\begin{equation}
\label{intervals_union}
	\left[ \theta^i_j, \theta^i_t \right)  =
    \bigcup\limits_{h = j + 1}^{t} \left[\theta^i_{h - 1}, \theta^i_h \right),
    \mbox{ where } t = j + 1, \ldots, l(i, j) ~.
\end{equation}

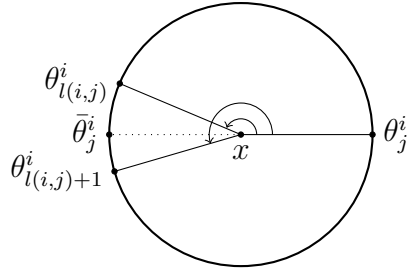
\begin{figure}[t]
\centering
	\begin{tikzpicture}[scale=0.7]
	\draw[thick] (0,0) circle (2.5cm);
    \filldraw (0,0) circle (0.05cm) coordinate(x) node[below] {$\x$};
    \filldraw (2.5, 0) circle (0.05cm) coordinate (a) node[right] {$\theta^i_j$};
    \filldraw (-2.3, 0.97) circle (0.05cm) coordinate (b) node[left] {$\theta^i_{l(i, j)}$};
    \filldraw (-2.5, 0) circle (0.05cm) coordinate (d) node[left] {$\bar{\theta}^i_j$};
    \filldraw (-2.4, -0.7) circle (0.05cm) coordinate (c) node[left] {$\theta^i_{l(i, j) + 1}$};
    \draw (2.5, 0) -- (0, 0) -- (-2.3, 0.97);
    \draw (0, 0) -- (-2.4, -0.7);
    \draw[->](x) +(0:.3cm) arc (0:154:.3cm);
    \draw[->](x) +(0:.6cm) arc (0:196:.6cm);
    \draw[dotted] (x) -- (d);
\end{tikzpicture}
\caption[ ]{Index $l(i, j)$ and index $\left(l(i, j) + 1\right)$} \label{fig:l(i,j) vs l(i,j)+1}
\end{figure}


\subsection{Computing $\Distar$}
\label{SS:3.2}
Let us the denote the count of the antipodes within the minor arc
between $a$ and $b$ by $c(a, b)$. Then $\Distar$ can be written as
follows:
\begin{equation} \label{initial_fla}
	\Distar = \sum\limits_{j = 0}^{n_i - 1} \sum\limits_{t = j + 1}^{l(i, j)} c \left(\theta^i_j, \theta^i_t \right) ~.
\end{equation}
Note that index $t$ is taken modulo $n_i$. From
(\ref{intervals_union}) we have:
\begin{equation} \label{counts_split}
	c \left( \theta^i_j, \theta^i_t \right) = \sum\limits_{h = j + 1}^t c \left( \theta^i_{h - 1}, \theta^i_h \right), \quad\mbox{for } t = j + 1,  \ldots, l(i, j) ~.
\end{equation}
Due to (\ref{initial_fla}) and (\ref{counts_split}), we have:
\begin{equation}
\label{consecutive_intervals}
	\Distar = \sum\limits_{j = 0}^{n_i - 1}\sum\limits_{t = j + 1}^{l(i, j)}\sum\limits_{h = j + 1}^t c \left( \theta^i_{h - 1}, \theta^i_h \right) ~.
\end{equation}
Let $C^i_{h} = c \left( \theta^i_{h - 1}, \theta^i_h \right)$, $\vert
C^i\vert = n_i$. Then (\ref{consecutive_intervals}) can be rewritten
as:
\begin{equation}
\label{substituted_f-la}
	\Distar = \sum\limits_{j = 0}^{n_i - 1}\sum\limits_{t = j + 1}^{l(i, j)}\sum\limits_{h = j + 1}^t C^i_{h} ~.
\end{equation}

Let us create an array of prefix sums: $S^i$, where
$S^i_t = \sum_{h = 0}^t{C^i_h}$, $\vert S^i\vert = n_i$. This array can
be filled in $O(n_i)$ time and proves to be very useful when we need to 
calculate a sum of the elements of $C^i$ between two certain indices. In 
fact, such sum can be obtained in constant time using the elements of
array $S^i$:

\begin{equation}
\label{sum_of_C}
	\sum_{h = j + 1}^t C^i_{h} =
    \begin{cases}
		S^i_t - S^i_j, & \mbox{if } t \geq j+1, j \neq n_i -1 ~, \\
        S^i_{n_i - 1} + S^i_t - S^i_j, & \mbox{if } t < j+1, j \neq n_i -1 ~, \\
        S^i_t, &\mbox{if } j=n_i-1 ~.
	\end{cases}
\end{equation}
Combining (\ref{substituted_f-la}) and (\ref{sum_of_C}), we get:
\begin{equation}
  \begin{split}
  \label{split_two}
\Distar = \sum\limits_{j = 0}^{n_i - 1}\sum\limits_{t = j + 1}^{l(i, j)} S^i_t - \sum\limits_{j = 0}^{n_i - 1} \left( (l(i, j) - j) \bmod{n_i} \right) \cdot S^i_j +
     \begin{cases}
		0, \mbox{ if } t \geq j + 1 ~ \mbox{or } ~ j = n_i -1 ~, \\
        \sum\limits_{j = 0}^{n_i - 1} \sum\limits_{t = j + 1}^{l(i, j)} S^i_{n_i - 1}, \mbox{if } t < j + 1 ~.
	 \end{cases}
  \end{split}
\end{equation}

Let us create another array of prefix sums $T^i$, where
$T^i_j = \sum_{t = 0}^j{S^i_t}$, $\vert T^i \vert = n_i$. This array is
used to retrieve the sum of elements of $S^i$ between the indices $j+1$ 
and $l(i, j)$ in $O(1)$ time:

\begin{equation}
	\sum_{t = j + 1}^{l(i, j)} S^i_t =
	\begin{cases}
		T^i_{l(i, j)} - T^i_j, & \mbox{if } l(i, j) \geq j + 1, j \neq n_i -1 ~, \\
        T^i_{n_i - 1} + T^i_{l(i, j)} - T^i_j, & \mbox{if } l(i, j) < j + 1, j \neq n_i -1 ~, \\
        T^i_{l(i, j)}, & \mbox{if } j = n_i-1.
	\end{cases}
\end{equation}
Also note that the index $t$ runs from $j+1$ to $l(i, j)$. So
$t < j + 1$ in (\ref{split_two}) is only possible if initially
$j + 1 > l(i, j)$ and we wrapped around the array. In other words,
$t < j + 1$ is equivalent to $j + 1 > l(i, j)$ and
$t = 0, \ldots, l(i, j)$. 
Hence: 

\begin{equation}
  \begin{split}
  \label{restated}  
	\Distar &= \sum\limits_{j = 0}^{n_i - 1} \left( T^i_{l(i, j)} - T^i_j
    -\left( (l(i, j) - j) \bmod{n_i} \right) \cdot S^i_j \right) \\
    &+ \begin{cases}
		n_i \cdot T^i_{n_i - 1} + \sum\limits_{j = 0}^{n_i - 1} \sum\limits_{t = 0}^{l(i, j)} S^i_{n_i - 1}, &\mbox{if } l(i, j) < j + 1 ~, \\
        0, &\mbox{otherwise} ~.
     \end{cases}
   \end{split}
\end{equation}
 
After simplifying, we obtain:
\begin{equation}
  \begin{split}
  \label{\Distar}  
     \Distar &= \sum\limits_{j = 0}^{n_i - 1} \left( T^i_{l(i, j)} - T^i_j
    -\left( (l(i, j) - j) \bmod{n_i} \right) \cdot S^i_j \right) \\
    &+ \begin{cases}
	n_i \cdot \left( T^i_{n_i - 1} + \left( \left( l(i, j) + 1 \right) \bmod{n_i} \right) \cdot S^i_{n_i - 1} \right), &\mbox{if } l(i, j) < j + 1 ~, \\
        0, &\mbox{otherwise}  ~.
     \end{cases} 
   \end{split}
\end{equation}

\subsection{Algorithm and Analysis}
\label{SS:3.3}

\algrenewcommand\Require{\textbf{Input: }}
\algrenewcommand\Ensure{\textbf{Output: }}

\begin{algorithm}
\caption[]{$\tt{CSD(\x, P)}$ \label{CSD}}
\tt{\Require{$\tt{\x, P = (P^1, \ldots, P^k)}$.}}
\Ensure{$\tt{\hatd(\x, P)}$.}
\begin{algorithmic}[1]
\State $\tt{Sum1 \gets 0, Sum2 \gets 0;}$
  \For{$\tt{i\gets 1, k}$}
  \For{$\tt{j\gets 0, n_i - 1}$}
    \State $\tt{\theta^i_j \gets}$ polar angle of $\tt{(P^i_j - \x) \bmod{ 2\pi};}$
    \State $\tt{\bar{\theta}^i_j \gets (\theta^i_j + \pi) \bmod{2\pi};}$
    \EndFor
    \State $\tt{Sort(\theta^i);}$
    \Comment{while permuting $\tt{\bar{\theta}^i}$}
    \State Restore the order in $\tt{\bar{\theta}^i;}$
    \State $\tt{Sum1 \gets Sum1 + D(\x, \theta^i);}$
    \Comment{use the algorithm from \cite{rousseeuw1996algorithm}}
  \EndFor
  \State $\tt{A \gets Merge(\bar{\theta}^1, \ldots, \bar{\theta}^k);}$
  \Comment{$\tt{A}$ is sorted}
  \State $\tt{D \gets D(\x, A);}$
  \Comment{use the algorithm from \cite{rousseeuw1996algorithm}}
  \For{$\tt{i\gets 1, k}$}
    \State $\tt{B \gets Merge (A, \theta^i);}$
    \Comment{update $\tt{p(\theta^i_j)}$ the pointers of $\tt{\theta^i_j}$,}\\
    \Comment{$\tt{B}$ stands for $\tt{A^i}$}
    \For{$\tt{j\gets 1, n_i}$}
        \Comment{$\tt{j = j \bmod{n_i}}$}
    	\If{$\tt{p(\theta^i_{j - 1}) < p(\theta^i_j)}$}
    	\State $\tt{C_j \gets p(\theta^i_j) - p(\theta^i_{j - 1}) - 1;}$
        \Comment{$\tt{C = C^i}$ - array of antipodal counts}
        \Else
        \State $\tt{C_j \gets n + n_i - p(\theta^i_{j - 1}) + p(\theta^i_j) - 1;}$
        \EndIf
      \EndFor
      \State Find $\tt{l(i, 0)}$ using binary search in $\tt{\theta^i;}$
      \State $\tt{S_0 \gets C_0}$; $\tt{T_0 \gets S_0;}$
      \Comment{$\tt{S = S^i, T = T^i}$ - prefix sum arrays}
      \For{$\tt{j\gets 1, n_i - 1}$}
      	\State Find $\tt{l(i, j);}$
        \State $\tt{S_j \gets S_{j - 1} + C_j;}$
        \State $\tt{T_j \gets T_{j - 1} + S_j;}$
      \EndFor
      \State $\tt{Sum2 \gets Sum2 + \Distar(\x, P)}$ obtained from the formula (\ref{\Distar});
      \State \textbf{delete} $\tt{B, C, S, T;}$
  \EndFor
\State \textbf{return} $\tt{\hatd(\x, P) = D - \left(Sum2 - 2 * Sum1\right);}$
\Comment{$\tt{Sum1} = \sum\limits_{i=1}^k{D(\x, P^i)}$}
\end{algorithmic}
\end{algorithm}

First, we find all polar angles and their antipodes, which takes 
$O(n)$ in total. Second, we sort the arrays of polar angles
$\theta^i$ and their corresponding antipodal
elements $\bar{\theta}^i$, which gives us $O\left(\sum\limits_{i=1}^k
n_i\log{n_i} \right)$. Third, we need to rotate $\bar{\theta}^i$, so
that they are in ascending order. This will take $O(n)$ time. Then we
compute for each $i$ the number of triangles with all three vertices 
of colour $i$ that contain $\x$, i.e.~$D(\x, P^i)$, 
using the algorithm of Rousseeuw and Ruts \cite{rousseeuw1996algorithm} 
for sorted data.
This will 
run in $O(n_i)$, for each $i$, or $O(n)$ in total. Hence lines 2-10
of the Algorithm \ref{CSD} take $O\left(\sum\limits_{i=1}^k n_i +
\sum\limits_{i=1}^k n_i\log{n_i}\right) = O(n\log{n})$ time to
complete. This follows from the facts that $\sum\limits_{i = 1}^k n_i
= n$ and $n \log{n}$ is convex.

To generate the sorted array $A$ of antipodes, we merge the $k$ single-coloured
arrays using a heap 
(following e.g.~\cite{DBLP:books/mg/CormenLR89}) in $O(n\log{k})$
time. We need to compute the monochrome depth $D(\x,P)$ of $\x$ with 
respect to all points in $P$, regardless of colour.  For this we can use 
the sorted array of antipodes 
rather than sorting the original array. Thus we again use the linear
time monochrome algorithm \cite{rousseeuw1996algorithm} with $\x$ and $A$. 
Note that working with the antipodes is equivalent due to the fact that
the  simplicial depth of $\x$ does not change if we rotate the system of
data points around the centre $\x$. 

After that, we execute a cycle of $k$ iterations -- one for each
colour. It starts with merging two sorted arrays $A$ and $\theta^i$,
which is linear in the size of arrays we are merging and takes
$O\left(\sum\limits_{i=1}^k \left(n + n_i\right)\right) =
O\left(kn\right)$ in total. Filling the arrays $C$ is linear. Since
the $l(i, j)$ appear in sequence in the array $\theta^i$, we find the
first one $l(i, 0)$ using a binary search that takes $O(\log{n_i})$, 
and $O(k\log{n})$ in total. Then we find the rest of $l(i, j)$ in
$O(n)$ time for each $i$ by scanning through the array starting from the
element $\theta^i_{l(i, 0)}$. The remaining the operations take constant
time to execute. Therefore, total running time of Algorithm \ref{CSD}
is $O\left(n + n\log{n} + n + n\log{k} + kn + k\log{n} + kn\right) 
 = O\left(n\log{n} + kn\right)$.  The $n\log{n}$ term corresponds to
the initial sorting of the data points, if they are presented in sorted
order, the running time drops to $O(kn)$.

As for space, arrays $\theta^i$, $\bar{\theta}^i$ and
$A$ take $O(3n) = O(n)$ space in total. Note that merging $k$ sorted 
arrays into $A$ can be done in place \cite{MR2650342}. At each 
iteration $i$, we create $B$ of size $O(n + n_i)$, and $C$, $S$, $T$
of size $O(n_i)$ each. Fortunately, we only need these arrays within
the $i^{th}$ iteration, so we can delete them in the end (line 31 of
the Algorithm \ref{CSD}) and reuse the space freed. To store the
indices $l(i, j)$, we need $O(n)$ space, which again can be reallocated
when $i$ changes. Thus the amount of
space used by our algorithm is $O(n)$. 

An implementation of this algorithm is available on-line \cite{CODE}.

\begin{remark}
\label{R:data_point_depth}
In Section~\ref{S:ccsm}, we will want to compute the colourful simplicial
depth of the data points themselves.  This can be done by computing 
$\hat{D}(\x, P \setminus \{x\})$ and counting colourful simplices
which have $\x$ as a vertex.  This is the number of pairs of vertices
of some other colour, so for $\x$ of colour $i'$, 
it is $\sum\limits_{\substack{i = 1\\i \neq i'}}
^k n_i\cdot\sum\limits_{\substack{j = i + 1\\j \neq i'}}^k n_j$.
A naive evaluation of this expression takes $O(k^2)$ time. Instead we 
use prefix sum arrays to compute it in linear time. Let $K_i = \sum\limits
_{\substack{j = 1\\j \neq i'}}^i n_j$. Then $\sum\limits_{\substack{j =
i + 1\\j \neq i'}}^k n_j = K_k - K_i$, and we compute 
$\sum\limits_{\substack{i = 1\\i \neq i'}}^k n_i \cdot \left(
K_k - K_i\right)$ to obtain the result. 
Array $K$ takes $O(k)$ space and takes linear time to fill out. 
\end{remark}

\section{Computing Colourful Simplicial Medians}
\label{S:ccsm}
\subsection{Preliminaries}
\label{SS:ccsm:1}

Consider a family of sets $P^1$, $P^2$, $\ldots$ , $P^k$ $\in \R^2$,
$k \geq 3$, where each $P^i$ consists of the points of some
particular colour $i$. Define $n_i = \vert P^i\vert$, for $i = 1,
\ldots, k$. Let $P$ be the union of all colour sets: $P =
\bigcup\limits_{i = 1}^k P^i$. Recall that we denote the CSD of a point 
$\x \in \R^2$ relative to $P$ by $\hat{D}(\x, P)$.

Our objective is to find a point $x$ inside the convex hull
of $P$, denoted $conv(P)$, maximizing $\hat{D}(x, P)$. Call the depth
of such a point $\hat{\mu}(P)$. 
Let $\S$ be the set of line segments formed by all possible pairs of
points $(A, B)$, where $A \in P^i$, $B \in P^j$, $i < j$. 
We will refer to these as \emph{colourful segments}. 

\begin{figure}[h]
\center
\begin{tikzpicture}[scale=.20]
    \coordinate(r1) at (20, 32);
    \coordinate(r2) at (24, 4);
    \coordinate(r3) at (24, 12);
    \coordinate(g1) at (4, 24);
    \coordinate(g2) at (16, 20);
    \coordinate(b1) at (32, 24);
    \coordinate(b2) at (8, 8);
    \draw[color=pink] (g1) -- (b1);
    \draw[color=pink] (g1) -- (r1) -- (b1);
    \draw[color=pink] (g1) -- (r2) -- (b1);
    \draw[color=pink] (g1) -- (r3) -- (b1);
    \draw[color=yellow] (g2) -- (b1);
    \draw[color=yellow] (g2) -- (r1) -- (b1);
    \draw[color=yellow] (g2) -- (r2) -- (b1);
    \draw[color=yellow] (g2) -- (r3) -- (b1);
    \draw[color=cyan] (g2) -- (b2);
    \draw[color=cyan] (b2) -- (r1) -- (g2);
    \draw[color=cyan] (b2) -- (r2) -- (g2);
    \draw[color=cyan] (b2) -- (r3) -- (g2);
    \draw[color=olive] (g1) -- (b2);
    \draw[color=olive] (g1) -- (r1) -- (b2);
    \draw[color=olive] (g1) -- (r2) -- (b2);
    \draw[color=olive] (g1) -- (r3) -- (b2);
    \filldraw[color=black, fill=red] (r1) circle (0.5cm) node[above] {$R_1$};
    \filldraw[color=black, fill=red] (r2) circle (0.5cm) node[right] {$R_2$};
    \filldraw[color=black, fill=red] (r3) circle (0.5cm) node[right] {$R_3$};
    \filldraw[color=black, fill=green] (3.5, 23.5) rectangle (4.5, 24.5)  node[left] {$G_1~$};
    \filldraw[color=black, fill=green] (15.5, 19.5) rectangle (16.5, 20.5)  node[left] {$G_2~$};
    \filldraw[color=black, fill=blue, rotate around={45:(b1)}] (31.5, 23.5) rectangle (32.5, 24.5)  node[right] {$B_1$};
    \filldraw[color=black, fill=blue, rotate around={45:(b2)}] (7.5, 7.5) rectangle (8.5, 8.5)  node[left] {$B_2~$};
    \filldraw (16, 24) circle (0.25cm) node[above] {$a~$};
    \filldraw (17.3, 24) circle (0.25cm) node[above] {$\quad ~b$};
    \filldraw (13.2, 18.4) circle (0.25cm) node[left] {$c$};
    \filldraw (12, 16) circle (0.25cm) node[left] {$e$};
    \filldraw (14.4, 17.7) circle (0.25cm)  node[right] {$\,d$};
    \filldraw (12.8, 15.2) circle (0.25cm) node[right] {$f$};
    \filldraw (18.2, 15.4) circle (0.25cm) node[below] {$g~$};
    \filldraw (17.6, 10.4) circle (0.25cm) node[below] {$h$};
    \filldraw (20.45, 11.1) circle (0.25cm) node[below] {$i~$};
\end{tikzpicture}
\caption[ ]{A configuration $P$ of 7 points in $\R^2$, whose simplicial
median has depth 6 and occurs at points $B_1, G_1, B_2, G_2, d, f, g$.}
\label{fig:median_in_color}
\end{figure}
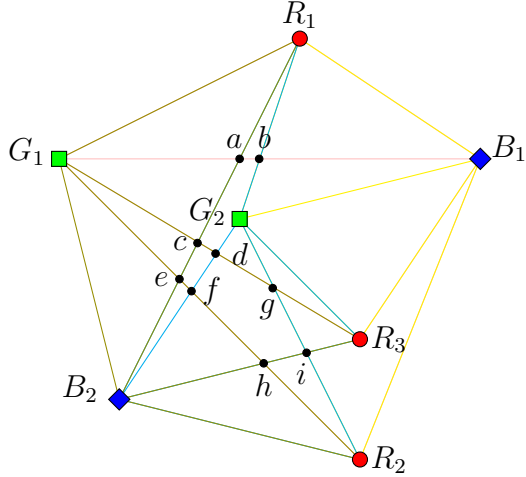

%
The following lemma (from~\cite{MR1989273}) is here adapted to a colourful setting:
\begin{lemma} \label{L:intersection_points}
To find a point with maximum colourful simplicial depth it suffices to
consider the intersection points of the colourful segments in $\S$.
\end{lemma}

\begin{proof}
The segments of $\S$ partition $conv(P)$ into cells\footnote{Unlike the
monochrome case, here some cells may not be convex, and some 
points of $conv(P)$ may fall outside any cell.}
of dimension
2, 1, 0 of constant colourful simplicial depth~\cite{MR2225675}.
Consider a 2-dimensional cell. Let
$p$ be a point in the interior of this cell, $q$ a point on the interior of
an edge and $v$ a vertex, so that $q$ and $v$
belong to the same line segment (Fig. \ref{fig:cell}). Then the
following inequality holds: $\hat{D} (p, P)\leq \hat{D}(q,
P)\leq \hat{D}(v, P)$, since any colourful simplices containing
$p$ also contain $q$, and any containing $q$ also contain $v$.
\end{proof}

\begin{figure}[ht]
\center
\begin{tikzpicture}[scale=0.5]
    \draw[-] (1, 1) -- (4, 7);
    \draw[-] (1, 5.5) -- (7, 4);
    \draw[-] (5, 6) -- (6, 3);
    \draw[-] (3, 1) -- (7, 5);
    \draw[-] (1, 4) -- (4, 1);
    \filldraw (4, 3) circle (0.05cm) coordinate (p) node[above] {$p$};
    \filldraw (2.5, 4) circle (0.05cm) coordinate (q) node[left] {$q$};
    \filldraw (3, 5) circle (0.05cm) coordinate (v) node[above] {$v$};
\end{tikzpicture}
\caption{An example of a cell} \label{fig:cell}
\end{figure}
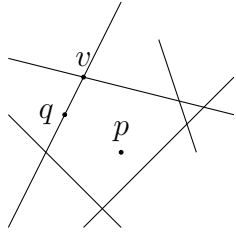


Let $col(A)$ denote the colour of a point $A$.
We store the segments in $\S$ as pairs of points: $\s = (A,
B)$, $col(A) < col(B)$. It is helpful to view each segment as 
directed, i.e.~a vector, with $A$ as
the tail and $B$ as the head. Each segment $\s$ extends to a
directed line $h$ dividing $\R^2$ into two open half-spaces:
$\s^+$ and $\s^-$, where $\s^+$ lies to the
right of the vector $\s$, and $\s^-$ to the left (Fig.
\ref{fig: \s^+ and \s^-}). We denote the set lines generated by
segments 
by $H$, so that every segment $\s \in \S$ corresponds to a line
$h \in H$.  

\begin{figure}[ht]
\centering
\subfigure[]{ 
\begin{tikzpicture}[scale=0.65]
	\filldraw (4, -1) circle (0.05cm) coordinate (c) node[below] 	{$A$};
    \filldraw (1, 2) circle (0.05cm) coordinate (d) node[left] {$B$};
	\draw[-{>}] (c) -- (d);
    \draw[draw=none] (3, 1.5) circle (0.05cm) coordinate (e) node[below] {$\s^+$};
    \draw[draw=none] (2, 0.5) circle (0.05cm) node[below] {$\s^-$};
\end{tikzpicture}
\label{fig:a}
} \hspace{2cm}
\subfigure[]{ 
\begin{tikzpicture}[scale=0.5]
	\filldraw (7, -2) circle (0.05cm) coordinate (a) node[below] {$A$};
    \filldraw (10, 2) circle (0.05cm) coordinate (b) node[right] {$B$};
	\draw[-{>}] (a) -- (b);
    \draw[draw=none] (9.5, 0) circle (0.05cm) node[below] {$\s^+$};
    \draw[draw=none] (7.5, .5) circle (0.05cm) node[below] {$\s^-$};
\end{tikzpicture}
\label{fig:b}
}
\caption{$\s^+$ and $\s^-$ of the segment $\s = (A, B)$} \label{fig: \s^+ and \s^-}
\end{figure}
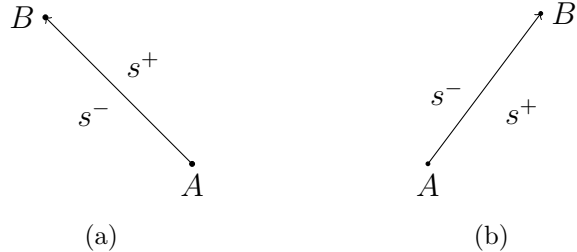 

We call the intersection points of the segments in $\S$
\textit{vertices}. Note that drawing the colourful segments
is equivalent to generating a rectilinear drawing of the complete 
graph $K_n$ with a few edges removed (the monochrome ones).
Thus, unless the points are concentrated in a single colour class,
the Crossing Lemma (see e.g.~\cite{PRTT06}) shows that we
will have $\Theta(n^4)$ vertices.  More precisely, we have a
rectilinear drawing of a complete $k$-partite graph; bounds
for this are considered some graphs from this family by
Gethner et al.~\cite{GHL+16} and references therein.

Computing the CSD of each of these points gives an 
$O(n^4 \log{n})$ algorithm for finding a simplicial median.
To improve this, we follow Aloupis et al.~\cite{MR1989273},
and compute the monochrome simplicial depth of most
vertices based on values of their neighbours and information
about the half-spaces of local segments.

Denote the number of points in $\s^+$ that have colours
different from the endpoints of $\s$ by $r(\s)$, and those in
$\s^-$ by $l(\s)$. Let $r^i(\s)$ and $l^i(\s)$ be the number of
points of a colour $i$ in $\s^+$ and $\s^-$ respectively. Let
$\bar{r}^i(\s)$ and $\bar{l}^i(\s)$ be the number of points of
all $k$ colours except for the colour $i$ in $\s^+$ and $\s^-$
respectively.  So for segment $\s = (A, B)$, we have
quantities as follows $\bar{r}^{col(A)}(\s) = \sum\limits_{
\substack{i = 1, \\ i \neq col(A)}}^k r^i(\s)$, $\bar{l}^{col(A)}(\s)
= \sum\limits_{\substack{i = 1, \\ i \neq col(A)}}^k l^i(\s)$.
Then it follows: $r(\s) = \bar{r}^{col(A)}(\s) - r^{col(B)}(\s)$ and
$l(\s) = \bar{l}^{col(A)}(\s) - l^{col(B)}(\s)$. The quantities
$\bar{r}^{col(A)}(\s)$ and $\bar{l}^{col(A)}(\s)$, $r^{col(B)}(\s)$
and $l^{col(B)}(\s)$, can be obtained as byproducts of an
algorithm that computes half-space depth.

The \textit{half-space depth} $HSD(\x, P)$ 
of a point $\x$ relative to data
set $P$ is the smallest number of data points in a half-plane
through the point $\x$ \cite{MR0426989}. 
An algorithm to compute half-space depth is
described by Rousseeuw and Ruts~\cite{rousseeuw1996algorithm},
it runs in $O(|P|)$ time when $P$ is sorted around $\x$.
It calculates the number of points $k_i$ in $P$ that lie strictly 
to the left of each line formed
by $\x$ and some point $P_i$, where $\x$ is the tail of the
vector $\overrightarrow{x P_i}$. 
Then the number of points to the right 
$\overrightarrow{x P_i}$ is $|P| - k_i - 1$.
Algorithm~\ref{alg:counts} uses a version of the half-space depth
algorithm to produce the quantities $r(\s)$ and $l(\s)$ that 
are used by our algorithm.

\algrenewcommand\Require{\textbf{\tt{Input: }}}
\algrenewcommand\Ensure{\textbf{\tt{Output: }}}
\algnewcommand\Or{\textbf{ or }}
\makeatletter\renewcommand{\ALG@name}{Algorithm}

\begin{algorithm}
\caption{\tt{Preprocessing: Computing $\tt{r(\s), l(\s)}$}}
\Require{$\tt{P, P^1, \ldots, P^k.}$}
\Ensure{$\tt{r(\s), l(\s)}$ \tt{for all} $\tt{\s \in \S}$.}
\begin{algorithmic}[1]
	\State \tt{Construct $\tt{List(P_i)}$ lists of points sorted
    around $\tt{P_i}$ for each $\tt{i = 0, \ldots, n - 1}$ \cite{MR812156};}
    \State $\tt{S \gets \emptyset}$, $\tt{H \gets \emptyset;}$
    \For{$\tt{i \gets 0, n-1}$}
    	\State $\tt{t \gets pop(List(P_i));}$
        \State $\tt{\bar{\theta}^{col(P_i)} =}$ \tt{polar angles of}
        $\tt{List(P_i)}$ \tt{with NO points of colour} $\tt{col(P_i);}$
        \State $\tt{\theta^{col(P_t)} =}$ \tt{polar angles of}
        $\tt{List(P_i)}$ \tt{of colour} $\tt{col(P_t)}$
        \tt{only;}
        \State \tt{Compute $\tt{\bar{r}^{col(P_i)}(\s),
        \bar{l}^{col(P_i)}(\s)}$ while running} $\tt{HSD(P_i, \bar{\theta}^{col(P_i)})}$        \cite{rousseeuw1996algorithm};
        \For{$\tt{i' \gets 1, k}$}
        	\If{$\tt{i' > col(P_i)}$}
            	\State \tt{Compute $\tt{r^{i'}(\s), l^{i'}(\s)}$ during the execution of} $\tt{HSD(P_i, \theta^{i'})}$
                \cite{rousseeuw1996algorithm};
                \For{$\tt{j \gets 0, n_{i'} - 1}$}
                	\State $\tt{\s \gets (P_i, P^{i'}_j);}$
                    \Comment{\tt{create a new segment}}
                    \State $\tt{ver(\s) \gets \emptyset;}$
                    \State $\tt{cross(ver(\s)) \gets \emptyset;}$
                    \State $\tt{h \gets (slope(\s), intercept(\s));}$
                    \State $\tt{push(\S, \s);}$
                    \State $\tt{push(H, h);}$
                    \State $\tt{r(\s) \gets \bar{r}^{col(P_i)}(\s) -
                    r^{i'}(\s);}$
                    \State $\tt{l(\s) \gets \bar{l}^{col(P_i)}(\s) -
                    l^{i'}(\s);}$
                \EndFor
                \State \tt{delete $\tt{\theta^{i'}}$};
            \EndIf
        \EndFor
        \State \tt{delete $\tt{\bar{\theta}^{col(P_i)}}$};
    \EndFor
\end{algorithmic}
\label{alg:counts}
\end{algorithm}

The algorithm of \cite{MR812156} will, for each $P_i \in P$, sort 
$P \setminus \lbrace P_i \rbrace$ around $P_i$ in $\Theta(|P|^2)$ time.
In particular, it assigns every point $P_i \in P$ a list of indices
that determine the order of points $P \setminus \lbrace P_i \rbrace$
in the clockwise ordering around $P_i$. Denote this by $List(P_i)$.
These ideas allow us to compute $r(\s)$ and $l(\s)$ for every
segment $\s$. At every iteration $i$, we form arrays of sorted polar
angles $\bar{\theta}^{col(P_i)}$ and $\theta^{i'}$. 
Together they take $O(2n) = O(n)$ space.

\subsection{Computing a Median}
\label{SS:ccsm:ts}
To compute the CSD of all vertices, we carry out a topological
sweep (see e.g.~\cite{MR990055}). 
We begin by extending the segments in $\S$ to a set of lines $H$.
The set $V^*$ of intersection points of lines of $H$ includes the
$\Theta(n^4)$ vertices $V$ which are on the interior of a pair of 
segments of $\S$, points from $P$, and additional exterior intersections.
We call points in $V^* \setminus V$ \textit{phantom vertices}.

Call a line segment
of any line in $H$ between two neighbouring vertices, or a ray
from a vertex on a line that contains no further vertices an \textit{edge}.  
A \textit{topological line} is a curve in $\R^2$ that is topologically 
a line and intersects each line in $H$ exactly once.  
We choose an initial topological line to be an unbounded curve 
that divides $\R^2$ into two pieces such that all the finitely many
vertices in $V$ lie on one side of the curve, by convention 
the right side.  We call this line the \textit{leftmost cut}.
We call a \textit{vertical cut} the list $(c_1, c_2,
\ldots, c_m)$ of the $m=|H|$ edges intersecting a particular topological line.
For each $i$, $1 \leq i \leq m - 1$, $c_i$ and
$c_{i + 1}$ share a 2-cell in the complex induced by $H$. 
Two vertical cuts are illustrated in Fig.~\ref{fig:leftmost_cut}. 

The topological sweep begins with the leftmost cut and moves across
the arrangement to the right, crossing one vertex at a time. If
two edges $c_i$ and $c_{i + 1}$ of the current cut have a common
right endpoint, we store the index $i$ in the stack $I$. For example,
in Figure \ref{fig:leftmost_cut}, $I = \lbrace 1, 4\rbrace$. An
\textit{elementary step} is performed when we move to a new vertex 
by popping the stack $I$. In Figure \ref{fig:top_sweep},
we have moved past the vertex $v$, a common right endpoint of $c_4$
and $c_5$ which is the intersection point of
$h_1$ and $h_2$. The updated stack is $I = \lbrace 1, 3\rbrace$.

\begin{figure}[ht]
\subfigure[The leftmost cut]{ 
\begin{tikzpicture}[scale=.35]
	\draw[draw=none] (1, 1) circle (0.05cm) coordinate (a1) node[left] {$h_1$};
    \draw[draw=none] (1, 2) circle (0.05cm) coordinate (a2) node[left] {$h_2$};
    \draw[draw=none] (1, 5) circle (0.05cm) coordinate (a3) node[left] {$h_3$};
    \draw[draw=none] (1, 8) circle (0.05cm) coordinate (a4) node[left] {$h_4$};
    \draw[draw=none] (1, 9) circle (0.05cm) coordinate (a5) node[left] {$h_5$};
    \draw[-] (1, 1) -- (16, 8);
    \draw[-] (1, 2) -- (16, 3);
    \draw[-] (1, 5) -- (16, 1);
    \draw[-] (1, 8) -- (16, 4);
    \draw[-] (1, 9) -- (16, 2);
    \filldraw (3.5, 2.16) circle (0.1cm) coordinate (c1) node[above] {};
    \filldraw (10, 2.6) circle (0.1cm) coordinate (c2) node[left] {};
    \filldraw (6.45, 3.54) circle (0.1cm) coordinate (c3) node[above] {};
    \filldraw (10.54, 5.46) circle (0.1cm) coordinate (c4) node[above] {};
    \filldraw (6, 6.66) circle (0.1cm) coordinate (c5) node[above] {};
    \filldraw (9.57, 5) circle (0.1cm) coordinate (c6) node[above] {};
    \filldraw (14.12, 2.87) circle (0.1cm) coordinate (c7) node[above] {};
    \draw[dashed] (2, 0.5) .. controls (2, 4) and (6, 7) .. (3.5, 9.5);
    \draw[line width=1] (a1) -- (c1);
    \draw[line width=1] (a2) -- (c1);
    \draw[line width=1] (a3) -- (c3);
    \draw[line width=1] (a4) -- (c5);
    \draw[line width=1] (a5) -- (c5);
    \draw[draw=none] (3, 8) circle (0.05cm) coordinate (a1) node[above] {$c_1$};
    \draw[draw=none] (2.5, 7.5) circle (0.05cm) coordinate (a1) node[below] {$c_2$};
    \draw[draw=none] (5, 4) circle (0.05cm) coordinate (a1) node[above] {$c_3$};
    \draw[draw=none] (1.5, 2) circle (0.05cm) coordinate (a1) node[above] {$c_4$};
    \draw[draw=none] (3, 1.5) circle (0.05cm) coordinate (a1) node[below] {$c_5$};
\end{tikzpicture}
\label{fig:leftmost_cut}
}
\subfigure[An elementary step in a topological sweep]{ 
\begin{tikzpicture}[scale=0.35]
	\draw[draw=none] (1, 1) circle (0.05cm) coordinate (a1) node[left] {$h_1$};
    \draw[draw=none] (1, 2) circle (0.05cm) coordinate (a2) node[left] {$h_2$};
    \draw[draw=none] (1, 5) circle (0.05cm) coordinate (a3) node[left] {$h_3$};
    \draw[draw=none] (1, 8) circle (0.05cm) coordinate (a4) node[left] {$h_4$};
    \draw[draw=none] (1, 9) circle (0.05cm) coordinate (a5) node[left] {$h_5$};
    \draw[-] (1, 1) -- (16, 8);
    \draw[-] (1, 2) -- (16, 3);
    \draw[-] (1, 5) -- (16, 1);
    \draw[-] (1, 8) -- (16, 4);
    \draw[-] (1, 9) -- (16, 2);
    \filldraw (3.5, 2.16) circle (0.1cm) coordinate (c1) node[above] {$v$};
    \filldraw (10, 2.6) circle (0.1cm) coordinate (c2) node[left] {};
    \filldraw (6.45, 3.54) circle (0.1cm) coordinate (c3) node[above] {};
    \filldraw (10.54, 5.46) circle (0.1cm) coordinate (c4) node[above] {};
    \filldraw (6, 6.66) circle (0.1cm) coordinate (c5) node[above] {};
    \filldraw (9.57, 5) circle (0.1cm) coordinate (c6) node[above] {};
    \filldraw (14.12, 2.87) circle (0.1cm) coordinate (c7) node[above] {};
    \draw[dashed] (4.5, 0.5) .. controls (6, 4) and (2, 7) .. (2.5, 9.5);
    \draw[line width=1] (c1) -- (c2);
    \draw[line width=1] (c1) -- (c3);
    \draw[line width=1] (a3) -- (c3);
    \draw[line width=1] (a4) -- (c5);
    \draw[line width=1] (a5) -- (c5);
    \draw[draw=none] (3.5, 8) circle (0.05cm) coordinate (a1) node[above] {$c_1$};
    \draw[draw=none] (2.3, 7.5) circle (0.05cm) coordinate (a1) node[below] {$c_2$};
    \draw[draw=none] (5, 4) circle (0.05cm) coordinate (a1) node[above] {$c_3$};
    \draw[draw=none] (5.7, 2) circle (0.05cm) coordinate (a1) node[above] {$c_6$};
    \draw[draw=none] (7, 2.3) circle (0.05cm) coordinate (a1) node[below] {$c_7$};
\end{tikzpicture}
\label{fig:top_sweep}
}
\caption{}
\end{figure}

We focus on the elementary steps, because at each step we can
compute the CSD of the crossed vertex.  As it moves, the topological 
line retains the property that everything
to the left of it has already been swept over. That is, if we are crossing
vertex $v$ that belongs to segment $\s$, every vertex of the line containing
$\s$ on the opposite side of the topological line prior to crossing has
already been swept.
For each segment $\s \in \S$ we
store the last processed vertex and denote it by $ver(\s)$, along
with its CSD. Since every vertex lies at the
intersection of two segments, we also store the crossing segment for
$\s$ and $ver(\s)$, denote it by $cross(ver(\s))$. Before starting the
topological sweep, for each $\s \in \S$ we assign $ver(\s) = \emptyset$,
and $cross(ver(\s)) = \emptyset$. After completing an elementary step
where we crossed a vertex $v$ that lies at the intersection of $\h_i$
and $\h_j$, we assign $ver(\h_i) \leftarrow v$, $ver(\h_j)
\leftarrow v$, $cross(ver(\h_i)) = \h_j$, $cross(ver(\h_j)) = \h_i$.

The topological sweep skips through phantom vertices (see 
Lemma~\ref{L:phantom}), and computes
the CSD of vertices in $P$ directly.
We now explain how we process a non-phantom
vertex $v$ at an elementary step when we have an adjacent vertex
already computed. Assume $v$ is
at the intersection of $\h_i = \overrightarrow{AB}$
and $\h_k = \overrightarrow{EF}$, see Figure \ref{fig:segments_coordinates}.
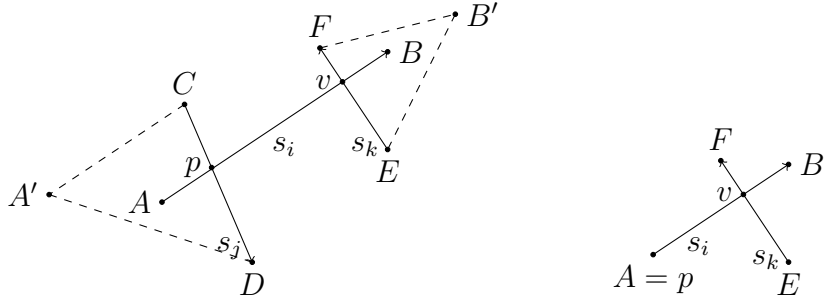
\begin{figure}[ht]
\centering
\subfigure[Two adjacent vertices $p$ and $v$ and their corresponding line
segments. A colourful triangle $\bigtriangleup CDA'$ contains $p$ but
not $v$, where $col(A') \notin \lbrace col(C), col(D)\rbrace$.
Similarly, a colourful triangle $\bigtriangleup EFB'$ contains $v$
but not $p$, where $col(B') \notin \lbrace col(E), col(F)\rbrace$.]{
	\begin{tikzpicture}[scale=.6]
    \filldraw (0.1, 0.1) circle (0.05cm) coordinate (p) node[left] {$p$};
    \filldraw (3, 2) circle (0.05cm) coordinate (v) node[left] {$v$};
    \filldraw (-1, -2/3) circle (0.05cm) coordinate (A) node[left] {$A$};
    \filldraw (-3.5, -0.5) circle (0.05cm) coordinate (A') node[left] {$A'$};
    \filldraw (4, 8/3) circle (0.05cm) coordinate (B) node[right] {$B$};
    \filldraw (5.5, 3.5) circle (0.05cm) coordinate (B') node[right] {$B'$};
    \filldraw (-.5, 1.5) circle (0.05cm) coordinate (C) node[above] {$C$};
    \filldraw (1, -2) circle (0.05cm) coordinate (D) node[below] {$D$};
    \filldraw (4, .5) circle (0.05cm) coordinate (E) node[below] {$E$};
    \filldraw (2.5, 2.75) circle (0.05cm) coordinate (F) node[above] {$F$};
    \draw[-{>}] (A) -- (B); 
    \draw[-{>}] (C) -- (D);
    \draw[-{>}] (E) -- (F);
    \draw[dashed] (C) -- (A');
    \draw[dashed] (A') -- (D);
    \draw[dashed] (F) -- (B');
    \draw[dashed] (B') -- (E);
    \draw[draw=none] (1.7, 1) circle (0.05cm) node[below] {$\h_i$};
    \draw[draw=none] (.5, -1.2) circle (0.05cm) node[below] {$\h_j$};
    \draw[draw=none] (3.5, 1) circle (0.05cm) node[below] {$\h_k$};
\end{tikzpicture}
\label{fig:segments_coordinates}
}
\qquad
\subfigure[Here $ver(\h_i) = \emptyset$, hence $cross(ver(\h_i)) =
\emptyset$, and we can not run Subroutine \ref{alg: subroutine}.]{
\begin{tikzpicture}[scale=.6]
    \filldraw (3, 2) circle (0.05cm) coordinate (v) node[left] {$v$};
    \filldraw (1, 2/3) circle (0.05cm) coordinate (A) node[below] {$A = p$};
    \filldraw (4, 8/3) circle (0.05cm) coordinate (B) node[right] {$B$};
    \filldraw (4, .5) circle (0.05cm) coordinate (E) node[below] {$E$};
    \filldraw (2.5, 2.75) circle (0.05cm) coordinate (F) node[above] {$F$};
    \draw[-{>}] (A) -- (B);
    \draw[-{>}] (E) -- (F);
    \draw[draw=none] (2, 1.3) circle (0.05cm) node[below] {$\h_i$};
    \draw[draw=none] (3.5, 1) circle (0.05cm) node[below] {$\h_k$};
\end{tikzpicture}
\label{fig:first_vertex}
}
\caption{Capturing a new vertex}
\end{figure}
Without loss of generality we
take $ver(\h_i) = p$, where $cross(ver(\h_i)) = \h_j$. We view this
elementary step as moving along the segment $\h_i$ from its
intersection point with $\h_j$ to the one with $\h_k$.
Each intersecting segment forms a triangle with every point strictly
to one side. Thus when we leave segment $\h_j = (C, D)$ behind,
we exit as many colourful triangles that contain $p$ as there are
points on the other side of $\h_j$ of colours different from
$col(C)$ and $col(D)$. When we encounter segment $\h_k = (E, F)$,
we enter the colourful triangles that contain $v$ formed by $\h_k$
and each point of a colour different from $col(E)$ and $col(F)$
on the other side of $\h_k$. Let us denote the $x$ and $y$
coordinates of a point $A$ by $A.x$ and $A.y$ respectively. Now, to
compute the CSD of $v$ knowing the CSD of $p$, we execute Subroutine
\ref{alg: subroutine}. 


\makeatletter\renewcommand{\ALG@name}{\tt{Subroutine}}

\begin{algorithm}
\caption{\tt{Computing $\hat{D}(v)$ from $\hat{D}(p)$}}
\Require{$\tt{\hat{D}(p), p, v, \h_j = (C, D), \h_k = (E, F).}$}
\Ensure{$\tt{\hat{D}(v).}$}
\begin{algorithmic}[1]
	\If{$\tt{(v.x - C.x)(D.y - C.y) - (v.y - C.y)(D.x - C.x) < 0}$}
      \State $\tt{\hat{D}(v) \gets \hat{D}(p) - r(\h_j);}$
    \Else 
      \State $\tt{\hat{D}(v) \gets \hat{D}(p) - l(\h_j);}$
    \EndIf
    \If{$\tt{(p.x - E.x)(F.y - E.y) - (p.y - E.y)(F.x - E.x) < 0}$}
      \State $\tt{\hat{D}(v) \gets \hat{D}(v) + r(\h_k);}$
    \Else 
      \State $\tt{\hat{D}(v) \gets \hat{D}(v) + l(\h_k);}$
    \EndIf
\end{algorithmic}
\label{alg: subroutine}
\end{algorithm}

When both $ver(\h_i) = \emptyset$, $ver(\h_k) = \emptyset$, i.e.
vertex $v$ is the first vertex to be discovered for both segments (Fig.
\ref{fig:first_vertex}), we 
execute $CSD(v, P)$ to find the depth, and otherwise update in the
usual way.  Since once a segment $\h$ has $ver(\h)$ nonempty it cannot
return to being empty, we call CSD at most $O(n^2)$ times.

\begin{lemma} \label{L:phantom}
At phantom vertices, we do not need to compute the CSD or update 
values of $ver(\h)$.
\end{lemma}
\begin{proof}
First, notice that the simplicial median itself won't occur at such
a point as it is either inside a cell or a colourful segment,
as noted in Lemma~\ref{L:intersection_points}.  Then, notice that
if we are moving along a segment of an extended line $\h$ outside
the segment (i.e.~extended from a segment, but not inside it) 
then we will next encounter either
another phantom vertex or a vertex from $P$.  In the first case, 
we repeat the current situation and no computations are needed,
and in the second, 
we do not require prior vertex information to find the depth
as we have already computed the vertex directly from a call to $CSD(v, P)$.

On the other hand, if we are moving along $\h$ inside a segment,
then the line we are crossing meets $\h$ outside of its segment.
This means that as we cross we are not entering or leaving any
colourful triangles, so colourful depth remains constant and will
not attain its maximum at the crossing.  We do not need to update
$ver(\h)$ since the previous point from $V$ remains the relevant
one from the perspective 
\end{proof}

Indeed, it is possible to run the algorithm without generating or
considering phantom vertices.  We include them here as they
provide a canonical starting point, the leftmost cut, and a clean
definition of vertical cuts.  Also, we wanted to address why they 
don't affect the calculation.
In practice, including phantom vertices will increase the time and 
memory requirements by constant factor.

\subsection{Running Time and Space Analysis}
\label{SS:M:3:2}
Algorithm \ref{alg: main} is our main algorithm. First, it 
computes the half-space counts $r(\s)$
and $l(\s)$, which has a running time of $O(n^2)$.
At the same time, we initialize the structure $\S$ that contains
the colourful segments, setting $ver(\s) = \emptyset$ and 
$cross(ver(\s))=\emptyset$  for all $\s \in \S$.
Note that these as well as $H$, $List(P_i)$,
$r(\s)$, $l(\s)$ require $O(n^2)$ storage. 

\algdef{SE}[SUBALG]{Indent}{EndIndent}{}{\algorithmicend\ }%
\algtext*{Indent}
\algtext*{EndIndent}
\makeatletter\renewcommand{\ALG@name}{Algorithm}

\begin{algorithm}
\caption{\tt{Computing $\tt{\hat{\mu}(P)}$}}
\Require{$\tt{P^1, \ldots, P^k, \S, H, r(\s), l(\s).}$}
\Ensure{$\tt{v, \hat{\mu}(P).}$}
\begin{algorithmic}[1]
	\State \tt{Run Algorithm \ref{alg:counts};}
	\Comment \tt{Compute $\tt{r(\s), l(\s)}$;}
    \State \tt{Sort $\tt{H}$ while permuting $\tt{\S}$;}
    \State $\tt{max \gets 0}$;
    \For{$\tt{i \gets 0, n - 1}$}
    	\State $\tt{\theta = }$ polar angles of $\tt{List (P_i);}$
    	\State $\tt{\hat{D}(P_i) \gets CSD(P_i, \theta);}$
        \If{d > max}
        	\State $\tt{max \gets \hat{D}(P_i);}$
            \State $\tt{median \gets P_i;}$
        \EndIf
    \EndFor
    \State \tt{I $\gets \emptyset$};
    \State \tt{Push common right endpoints of the edges of the leftmost cut onto I};
    \While{$I \neq \emptyset$}
    \Comment{\tt{Start of the topological sweep.}}
    \State $\tt{v \gets pop(I)};$
    \Comment{$\tt{v}$ lies at the intersection of $\tt{\s_i = (A, B)}$ and $\tt{\s_k = (E, F)}$}
    \If{\tt{$v$ lies in the interiors of $\s_i$ and $\s_k$}}
    \If{$\tt{ver(\s_i) = \emptyset \And ver(\s_k) = \emptyset}$}
    	\State $\tt{\hat{D}(v) = CSD(v, P);}$
    \ElsIf{$ver(\s_i) \neq \emptyset$}
    	\State $\tt{\hat{D}(v) \gets}$ \tt{Subroutine \ref{alg: subroutine}} $\tt{(\hat{D}(p), p, v, \s_j, \s_k);}$
    	\Comment{\parbox[t]{.35\linewidth}{$\tt{p = ver(\s_i), \s_j = cross(ver(\s_i))}$}}
    \Else
    	\State $\tt{\hat{D}(v) \gets}$ \tt{Subroutine \ref{alg: subroutine}} $\tt{(\hat{D}(p), p, v, \s_j, \s_i);}$
    	\Comment{\parbox[t]{.35\linewidth}{$\tt{p = ver(\s_k), \s_j = cross(ver(\s_k))}$}}
    \EndIf
    \If{$\tt{\hat{D}(v) > max}$}
          \State $\tt{max \gets \hat{D}(v);}$
          \State $\tt{median \gets v;}$
    \EndIf
    \State $\tt{ver(\s_i) \gets v}$, $\tt{ver(\s_k) \gets v}$,
    $\tt{cross(ver(\s_i)) \gets \s_k}$, $\tt{cross(ver(\s_k)) \gets \s_i;}$
    \EndIf
    \State \tt{Push any new common right endpoints of the edges onto I};
    \EndWhile
    \Comment{\tt{End of the topological sweep.}}
    \State \Return $\tt{(median, max).}$
\end{algorithmic}
\label{alg: main}
\end{algorithm}

Sorting the lines in $H$ according to their slopes while also
permuting the segments in $\S$ takes $O(n^2\log{n})$ time. We assume
non-degeneracy and no vertical lines (these can use some special
handling, see e.g. \cite{MR1213465}). Computing the CSD of points
where no previous vertex is available takes $O(n^2 \log{n} + kn^2)$ total time.
The topological sweep takes linear time in the number of intersection points 
of H, so $O(n^4)$. 
We do not store all the vertices,
but only one per segment. Steps 15--29 (except for 18) in Algorithm
\ref{alg: main} take O(1) time, including the calls to the Subroutine
\ref{alg: subroutine}. As for step 18, it is executed at most $n^2$
times; following its execution, $ver(\s)$ will be initiated for the
two relevant segments.
Therefore, the total time it will take is $O(n^3\log{n} + kn^2)$.
Hence overall our algorithm takes $O(n^4)$ time and needs $O(n^2)$
storage.

Algorithm \ref{alg: main} returns a point that has maximum colourful
simplicial depth along with its CSD. It is simple to modify the algorithm
to return a list of all such points if there is more than one. 
We believe that maintaining such a list will not increase the required
storage, i.e.~it will contain $O(n^2)$ points throughout the execution
of the algorithm, but we haven't proved this.

\section{Conclusions and Questions}
\label{S:Concl}
Our main result is an algorithm computing the colourful simplicial depth
of a point $\x$
relative to a configuration $P = \left(P^1, P^2, \ldots, P^k\right)$
of $n$ points in $\R^2$ in $k$ colour classes can be solved in
$O(n \log{n} + kn)$ time, or in $O(kn)$ time if the input is sorted.
If we assume, as seems likely, that we cannot do better without
sorting the input, then for fixed $k$ this result is optimal up
to a constant factor.  It is an interesting question whether we
can improve the dependence on $k$, in particular when $k$ is large.

Computing colourful simplicial depth in higher dimension is very
challenging, in particular because there is no longer a natural
(circular) order of the points.  
Non-trivial algorithms for monochrome depth do exist
in dimension 3 \cite{CO01}, \cite{MR1189827}, but we do not know 
of any non-trivial algorithms for $d \ge 4$.
Algorithms for monochrome and colourful depth in higher dimension 
are an appealing challenge.
Indeed, for $(d+1)$ colours in $\R^d$, it is not even clear how 
efficiently one can exhibit a single colourful simplex containing
a given point~\cite{BO97}, \cite{DHST08}.

\section*{Acknowledgments}
This research was partially supported by an NSERC Discovery Grant to
T.~Stephen and by an SFU Graduate Fellowships to O.~Zasenko.
We thank A.~Deza for comments on the presentation.
\bibliographystyle{amsalpha} 
\bibliography{references.bib}

\end{document}